\documentclass[review]{elsarticle}

\usepackage{amsthm}
\theoremstyle{plain}
\usepackage{amsmath}
\usepackage{xcolor}
\usepackage{listings}
\usepackage[colorlinks,linkcolor=gray]{hyperref}
\usepackage[justification=centering]{caption}
\usepackage{multirow}
\usepackage{adjustbox}
\usepackage{threeparttable}
\usepackage{float}
\usepackage{wrapfig}
\usepackage{rotating}
\usepackage[normalem]{ulem}

\journal{Journal of \LaTeX\ Templates}

\newtheorem{fact}{Fact}
\newtheorem{definition}{Definition}
\newtheorem{theorem}{Theorem}

\begin{document}
\begin{frontmatter}
\title{Quantum Speedup and Limitations on Matroid Property Problems}

\author{Xiaowei Huang, Jingquan Luo, Lvzhou Li\footnote{L. Li is the corresponding author (lilvzh@mail.sysu.edu.cn)}}
\address{Institute of Quantum Computing and Computer Theory,  School of Computer and Engineering, Sun Yat-sen University, Guangzhou 510006, China}

\begin{abstract}
This paper initiates the study of  quantum algorithms for matroid property problems.
It is shown that  quadratic quantum speedup  is possible for the calculation
problem of finding the girth or the number of
circuits (bases, flats, hyperplanes) of a matroid, and for the  decision problem of
deciding whether a matroid is uniform or Eulerian, by 
giving a uniform lower bound $\Omega(\sqrt{\binom{n}{\lfloor n/2\rfloor}})$ on the 
query complexity for all these problems. On the other hand, for the uniform matroid 
decision problem,  an asymptotically optimal quantum algorithm is proposed which achieves
the lower bound, and  for the girth problem,  an almost optimal quantum algorithm is given
with query complexity $O(\log n\sqrt{\binom{n}{\lfloor n/2\rfloor}})$. In addition, for 
the paving matroid decision problem,  a lower bound 
$\Omega(\sqrt{\binom{n}{\lfloor n/2\rfloor}/n})$ on the query complexity is obtained, and  an
$O(\sqrt{\binom{n}{\lfloor n/2\rfloor}})$ quantum algorithm is presented.
\end{abstract}

\end{frontmatter}


\section{Introduction}\label{Int}
\subsection{Background}

The concept of matroids was originally introduced by Whitney  \cite{JHUP/whitney35}
in 1935 as a generalization of the concepts of linear spaces and graphs. After eighty years of
 development, the theory of matroids  has become an important branch of mathematics. It has already
become an effective mathematical tool to study other mathematic branches and has many
applications in geometry, topology, network theory and coding theory, especially in
combinatorial optimization \cite{DBLP:journals/mp/Edmonds71,HRW/lawler76,
  DBLP:journals/comsur/BassoliMRST13,DBLP:conf/ismp/Iri82}. 
When people want to study the commonality of some  problems from a more abstract level, matroids become the focus of this kind of research, because many important problems can be regarded as instances of matroid problems. For example, finding the maximum matching of a bipartite graph is essentially the problem of finding the intersection of two matroids, the coloring problem is essentially a matroid partition problem, and finding the Hamiltonian circuit of a graph is essentially finding the intersection of three matroids. It is because of this high degree of generality that the matroid theory has received a lot of attention in mathematics and computer science since the 1950s. 

At the same time, with the rapid development of quantum computing, finding more problems that can take advantage of quantum speedup has become one of the focus issues in the field of quantum computing. If quantum algorithms with speedup advantages can be obtained for some basic matroid problems, it will bring quantum fast solutions to a series of specific application problems. Therefore, it is very interesting to explore what matroid problems can be accelerated by quantum computing and to discover the inherent quantum complexity of these problems.

However, currently there has been no work considering quantum algorithms for matroid problems. There is also little research linking the two terms of quantum computing and matroids. Kulkarni and Santha \cite{DBLP:conf/ciac/KulkarniS13}  discussed the quantum query complexity for computing the characteristic function of a matroid, but this work  involves no quantum algorithm for any basic matroid problems. Of course, there exist  papers discussing quantum algorithms for  some instances of matroid problems, since as mentioned before,  matroids are a very generalized concept. For example, Refs.
\cite{Aghaei2008AQA,quinn2020solving} discussed  quantum algorithms for finding the minimum spanning tree of a graph, which  can be regarded as an instance of finding a basis set  of a graph matroid with the least possible total weight. Refs.
\cite{DBLP:conf/stacs/AmbainisS06,DBLP:journals/mst/Dorn09} studied  quantum algorithms for finding the maximum matching of bipartite graphs, which  can be regarded as an instance of finding the intersection of two matroids. In addition, matroids are a special case of submodular functions, and quantum algorithms for optimization of submodular functions has been discussed \cite{DBLP:journals/qic/HamoudiRRS19}. However, it is worth pointing out that all these works mentioned above do not directly examine the problem from the perspective of matroids.
By the way, 
there is also few work considering the application of the matroid theory to quantum computing. In 2014, Amy, Maslov and Mosca
\cite{amy2014polynomial} proposed to optimize the quantum circuit composed of CNOT gates and $T$ gates based on the matroid partitioning algorithm. More recently, Mann\cite{DBLP:journals/corr/abs-2101-00211} established a classical heuristic algorithm to accurately calculate quantum  amplitudes, which maps the output  amplitudes of quantum circuits to the Tutte polynomial estimation problem of graph matroids. 

\subsection{Our contributions}
As can be seen from the above, there is still very little research on the intersection of matroids and quantum computing. Especially considering that matroid problems are generalizations of many important specific problems, it is very meaningful to explore whether quantum computing has the advantage of speeding up the solution of matroid problems.  Thus, in this paper we try to explore the possibility 
 of quantum speedup on some basic matroid problems. A matroid  is a tuple $\mathcal{M}=(V,\mathcal{I})$ where  $V$ is a finite ground set and
$\mathcal{I}\subseteq 2^{V}$  satisfies three conditions (See Definition  \ref{definition4matroid}).
 It is generally assumed that a matroid can only
 be accessed through the {\it independence  oracle} $\mathcal{O}_{i}$: given a matroid $\mathcal{M}=(V,\mathcal{I})$, for a subset $S\subseteq V$, 
$\mathcal{O}_{i}(S)=1$ iff $S\in \mathcal{I}$. An algorithm for matroid problems should query the oracle as least as possible.  We study how well quantum query 
 algorithms perform on the matroid property problems (in Table  \ref{Tab:complexity}). Given a matroid  $\mathcal{M}=(V,\mathcal{I})$ with $|V|=n$, 
 our main results are as follows.
 
\begin{enumerate}
\item For the calculation problem  of finding the girth or  the number of circuits (bases, flats, hyperplanes), a quantum algorithm has to query the independence oracle at least
  $\Omega(\sqrt{\binom{n}{\lfloor n/2 \rfloor}})$ times (
  see Theorem \ref{theorem4nopolymoialalgorithm}).
\item For the decision problem of deciding whether a matroid is uniform or Eulerian, a quantum algorithm has to query the independence oracle at least
  $\Omega(\sqrt{\binom{n}{\lfloor n/2 \rfloor}})$ times  (see Theorem \ref{theorem4nopolymoialalgorithm2}).
\item For the uniform matroid decision problem, there is an
  $O(\sqrt{\binom{n}{\lfloor n/2\rfloor}})$ quantum algorithm which is asymptotically
  optimal (see Theorem \ref{theorem4uniformmatroidupperbound}), and for the  problem of finding the girth, there is an 
  $O(\log n\sqrt{\binom{n}{\lfloor n/2\rfloor}})$ quantum algorithm 
  which is almost optimal (see Theorem \ref{theorem4findgirth}).
\item For the paving matroid decision problem, there is a quantum algorithm using
  $O(\sqrt{\binom{n}{\lfloor n/2\rfloor}})$ queries (see Theorem
  \ref{theorem4pavingmatroidupperbound}) and any quantum algorithm has to call at least
  $\Omega(\sqrt{\binom{n}{\lfloor n/2\rfloor}/n})$ queries
  (see Theorem \ref{theorem4pavingmatroidlowerbound}).
\end{enumerate}

\begin{center}
  \begin{table}[htb]
    \caption {The query complexity of matroid property problems, where the quantum bounds are obtained in this paper.
    CL: Classical Lower Bound. QL: Quantum Lower Bound. QU: Quantum Upper Bound.}
    \label{Tab:complexity}
    \renewcommand{\arraystretch}{1.6}
    \begin{adjustbox}{width=\textwidth}
      \begin{tabular}{|c|l|c|c|c|}      
        \hline
        Type & Matroid Property Problems& 
        \textbf{CL}& \textbf{QL}& \textbf{QU}\\
        \hline
        \multirow{5}{*}{\rotatebox{270}{Calculation Problems}}& 1.Find the girth of $\mathcal{M}$.&
        \multirow{5}{*}{$\Omega(\binom{n}{\lfloor n/2\rfloor}$}&
        \multirow{5}{*}{$\Omega(\sqrt{\binom{n}{\lfloor n/2\rfloor}})$}&
        $O(\log n\sqrt{\binom{n}{\lfloor n/2\rfloor}})$\\
        &2.Find the number of circuits of $\mathcal{M}$.&
        && ---\\
        &3.Find the number of bases of $\mathcal{M}$.&
        && ---\\
        &4.Find the number of flats of $\mathcal{M}$.&
        && ---\\
        &5.Find the number of hyperplanes of $\mathcal{M}$.&
        && ---\\
        \hline
        \multirow{5}{*}{\rotatebox{270}{Decision Problems}}& 6.Is $\mathcal{M}$ an uniform matroid?&
        \multirow{2}{*}{$\Omega(\binom{n}{\lfloor n/2\rfloor}$}&         
        \multirow{2}{*}{$\Omega(\sqrt{\binom{n}{\lfloor n/2\rfloor}})$}&
        $O(\sqrt{\binom{n}{\lfloor n/2\rfloor}})$\\
        &7.Is $\mathcal{M}$ an Eulerian matroid?&
        && ---\\ \cline{3-3}
        \cline{4-4}
        &8.Is $\mathcal{M}$ a paving matroid?& $\Omega(\binom{n}{\lfloor n/2\rfloor}/n)$
        &$\Omega(\sqrt{\binom{n}{\lfloor n/2\rfloor}/n})$&
        $O(\sqrt{\binom{n}{\lfloor n/2\rfloor}})$\\
        \cline{2-5}
        &9.Is $\mathcal{M}$ a trivial matroid?&\multirow{2}{*}{$\Omega(n)$}
        &\multirow{2}{*}{$\Omega(\sqrt{n})$}&\multirow{2}{*}{$O(\sqrt{n})$}\\
        &10.Is $\mathcal{M}$ a loopless matroid?&
        &&\\
        \hline
      \end{tabular}        
      \renewcommand{\arraystretch}{1}
    \end{adjustbox}
    \begin{tablenotes}
    \item[*] 
    \end{tablenotes}
  \end{table}
\end{center}

The remainder of this paper is organized as follows.
In Section 2, some basic concepts in matroid theorey and quantum query model are
introduced.
In Section 3, we obtain lower bounds for some matroid properties problems and 
present quantum algorithms for some of these problems.
Section 4 concludes this paper and presents some further problems.

\section{Preliminaries}\label{Pre}
\subsection{Matroid Theory}

In this subsection, we will give some basic definitions  used in this paper. Since 
matroid theory was established as a generalization of linear algebra and graph theory,
many concepts in matroid theory are derived from these two disciplines. So if one is
familiar with linear algebra and graph theory, it  will be helpful for
understanding the following definitions. One can refer to \cite{OUP/oxley11} for more details about matroid theory.

\begin{definition}[\textbf{Matroid}]
  \label{definition4matroid}    
  A \emph{matroid} is a combinational object defined by the tuple
  $\mathcal{M}=(V,\mathcal{I})$ for finite set $V$ and $\mathcal{I}\subseteq 2^V$ such
  that the following properties hold:
  \begin{enumerate}
  \item[\textbf{I0}.] $\emptyset \in \mathcal{I}$;
  \item[\textbf{I1}.] If $A'\subseteq A$ and $A\in\mathcal{I}$, then $A'\in\mathcal{I}$;
  \item[\textbf{I2}.] For any two sets $A,B\in\mathcal{I}$ with $|A|<|B|$, there exist an
    element $v\in B-A$ such that $A\cup\{v\}\in\mathcal{I}$.
  \end{enumerate}
\end{definition}

\begin{definition}[\textbf{Independent Set}]
  For a matroid $\mathcal{M}=(V,\mathcal{I})$, we call $S\subseteq V$ \emph{independent}
  if $S\in\mathcal{I}$ and \emph{dependent} otherwise.
\end{definition}

\begin{definition}[\textbf{Circuit}]
  For a matroid $\mathcal{M}=(V,\mathcal{I})$, if a dependent set $C\subseteq V$ satisfies
  that for any $e\in C$, $C-e\in\mathcal{I}$, we call $C$ be a \emph{circuit} of
  $\mathcal{M}$. If $C=\{e\}$ is a circuit, we call $e$ be a \emph{loop}. If
  $C=\{e_1,e_2\}$ is a circuit, we call $e_1,e_2$ are \emph{parallel}.
\end{definition}

\begin{definition}[\textbf{Girth}]
  For a matroid $\mathcal{M}$, the \emph{girth} $g(\mathcal{M})$ of $\mathcal{M}$ is the
  minimum circuit size of $\mathcal{M}$ unless $\mathcal{M}$ has no circuits, in which
  case, $g(\mathcal{M})=\infty$.
\end{definition}

\begin{definition}[\textbf{Rank}]
  For a matroid $\mathcal{M}=(V,\mathcal{I})$, we define the \emph{rank} of $\mathcal{M}$
  as $\text{rank}(\mathcal{M})=\max_{S\in\mathcal{I}}|S|$. Further, for any $S\subseteq V$ we
  define $\text{rank}_\mathcal{M}(S)\equiv\max_{T\subseteq S:T\in\mathcal{I}}|T|$, for
  simplicity we use rank($S$) or $r(S)$ in place of rank$_\mathcal{M}$($S$) if there is
  no unambiguous in the context.
\end{definition}

\begin{definition}[\textbf{Base}]
  For a matroid $\mathcal{M}=(V,\mathcal{I})$, if $B\in\mathcal{I}$ such that
  $\text{rank}(B)=\text{rank}(\mathcal{M})$, we call $B$  a \emph{base} of $\mathcal{M}$. By the
  matroid's property (\textbf{I2}), we can see that if $B_1$ and $B_2$ are two distinct
  bases of $\mathcal{M}$, then $|B_1|=|B_2|$. Furthermore, if $e$ is any element of $B_1$,
  then there is an element $f\in B_2$ such that $(B_1 - \{e\})\cup\{f\}$ is also a base.
  A matroid can also be defined by a set of bases, which is equivalent to Definition
  \ref{definition4matroid}.    
\end{definition}

\begin{definition}[\textbf{Free Matroid and Trivial Matroid}]
  \label{definition4freeandtrivialmatroid}
  For a matroid $\mathcal{M}=(V,\mathcal{I})$, $\mathcal{M}$ is called a \emph{free}
  matroid if $V$ is the only base and a \emph{trivial} matroid if $\emptyset$ is the only
  base.
\end{definition}

\begin{definition}[\textbf{Loopless Matroid}]
  \label{definition4looplessmatroid}
  For a matroid $\mathcal{M}=(V,\mathcal{I})$, $\mathcal{M}$ is called a \emph{loopless}
  matroid if all the singleton in $V$ are independent. In other words, $\mathcal{M}$ does
  not have any circuit with size 1.
\end{definition}

\begin{definition}[\textbf{Uniform Matroid}]
  For a matroid $\mathcal{M}=(V,\mathcal{I})$, let $|V|=n$. If there exists an integer
  $r$ with $0\leq r \leq n$ such that $\mathcal{I}=\{S\subseteq V:|S|\leq r\}$, 
  $\mathcal{M}$ is called an \emph{uniform matroid} of rank $r$ and denoted by $U_{r,n}$.
\end{definition}

\begin{definition}[\textbf{Paving Matroid}]
  For a matroid $\mathcal{M}=(V,\mathcal{I})$, if every circuit $C$ of $\mathcal{M}$
  satisfies that $|C|\geq \text{rank}(\mathcal{M})$, we call $\mathcal{M}$ a
  \emph{paving matroid}. Obviously, an uniform matroid is also a paving matroid.
\end{definition}

\begin{definition}[\textbf{Closure}]
  Given a matroid $\mathcal{M}=(V,\mathcal{I})$ on ground set $V$ with rank function
  $r$. Let cl be the function from $2^V$ to $2^V$ defined, for all $X\subseteq V$, by
  $\text{cl}(X)=\{x\in V:r(X\cup x)=r(X)\}$.
  This function is called the \emph{closure operator} of $\mathcal{M}$, and we call
  cl($X$) the \emph{closure} or \emph{span} of $X$ in $\mathcal{M}$.
\end{definition}

\begin{definition}[\textbf{Flat and Hyperplane}]
  Given a matroid $\mathcal{M}=(V,\mathcal{I})$ on ground set $V$ and its closure
  operator cl, a subset $X$ of $V$ for which $\text{cl}(X) = X$ is called a
  \emph{flat} or a \emph{closed set} of $\mathcal{M}$. A \emph{hyperplane} of
  $\mathcal{M}$ is a flat of rank $r(\mathcal{M})-1$. A subset $X$ of $V$ is a
  \emph{spanning set} of $\mathcal{M}$ if $\text{cl}(X) = V$. We also say that $X$ spans
  a subset $Y$ of $V$ if $Y\subseteq \text{cl}(X)$.
\end{definition}

\begin{definition}[\textbf{Isomorphic Matroids}]
  Two matroids
  $\mathcal{M}_1=(V_1,\mathcal{I}_1)$ and $\mathcal{M}_2=(V_2,\mathcal{I}_2)$ are
  \emph{isomorphic}, written $\mathcal{M}_1\cong\mathcal{M}_2$, if there is a bijection
  $\psi$ from $V_1$ to  $V_2$ such that, for all $X\subseteq V_1$, the set $\psi(X)$
  is independent in $\mathcal{M}_2$ if and only if $X$ is independent in
  $\mathcal{M}_1$. We call such a bijection $\psi$ an \emph{isomorphism} from
  $\mathcal{M}_1$ to $\mathcal{M}_2$.
\end{definition}

\begin{definition}[\textbf{Eulerian Matroid}]
  For a matroid  $\mathcal{M}=(V,\mathcal{I})$, it is called an \emph{Eulerian matroid}
  if there exist disjoint circuits $C_1,\cdots,C_p$ such that $V=C_1\cup\cdots\cup C_p$.
\end{definition}

\textbf{Notations\label{matroidnotation}.}
We will always denote by $\mathcal{M}$ a matroid $(V,\mathcal{I})$ on a finite ground set
$V$ with $\mathcal{I}\subseteq 2^V$.
Given a matroid $\mathcal{M}$, we will denote the ground set and the set of
independent sets of $\mathcal{M}$ by $V(\mathcal{M})$ and $\mathcal{I}(\mathcal{M})$
respectively. Similarly, $\mathcal{C}(\mathcal{M})$, $\mathcal{B}(\mathcal{M})$,
$\mathcal{F}(\mathcal{M})$, $\mathcal{H}(\mathcal{M})$,
is the set of circuits, bases, flats, hyperplanes of $\mathcal{M}$, respectively.
$\text{cl}_\mathcal{M}$ is the closure operator of $\mathcal{M}$.
A set having $r$ elements will be call a $r$-set. $[n]$ denotes $\{1,2,\cdots,n\}$.
Given a ground set $V$ with $|V|=n$, $A\subseteq V$ and an integer $1\leq r\leq n$,
$J_{r}=\{J\subseteq V:|J|=r\}$, $A_{e}= A\cup\{e\}$ for any $e\in V-A$,
$J_{A}=\{A_{e}:e\in V-A\}$.\\

\textbf{Matroid Representation\label{matroidrepresentation}.}
For the convenience of in the following text, we use a $2^{n}$-bit $0-1$ string to
represent a matroid on the ground set $V$ with $|V|=n$ and each bit represents a subset
of $V$. A matroid $\mathcal{M}$ on the ground set $V$, denoted by $x(\mathcal{M})$ with
$x(\mathcal{M})\in\{0,1\}^{2^{n}}$. Similarly, we use $\mathcal{M}(x)$ to denote the
matroid determined by a $x\in\{0,1\}^{2^n}$ if $x$ encodes a matroid.
For any $i\in [2^n]$, $x_{i}=1$ indicates that the subset corresponding to
the $i$-th bit is an independent set of $\mathcal{M}$, otherwise $x_{i}=0$.
If we know the rank $r$ of a matroid on the ground set $V$, we usually use
a $\binom{n}{r}$-bit $0-1$ string to represents the matroid, each bit representing a
$r$-set which is set to 1 when the $r$-set is a base of the matroid,
otherwise is set to 0. Given a 0-1 string, which represents a matroid, it is easy to know it is a subset representation or a $r$-set representation from the context.\\

\textbf{Independence Oracle.} Given a matroid $\mathcal{M}=(V,\mathcal{I})$, assume we can only access it by querying the independence
oracle $\mathcal{O}_{i}$. For a subset $S\subseteq V$,
$\mathcal{O}_{i}(S)=1$ if $S$ is an independent set of $\mathcal{M}$, otherwise
$\mathcal{O}_{i}(S)=0$.

\subsection{Quantum Computation}
For the basic concepts and notations on quantum computing, we refer the reader to the
textbook by Nielsen and Chuang \cite{CUP/nielsen10}.
Throughout this paper, we use two basic tools: (i) Grover's algorithm and (ii) Ambainis's
quantum adversary method.
\subsubsection{Quantum Query Model}
In the quantum query model\cite{DBLP:conf/focs/BealsBCMW98}, the input bits of a boolean function
$f:\{0,1\}^N\rightarrow\{0,1\}$ can be accessed by queries to an oracle
$\mathcal{O}$. We use $\mathcal{O}_x$ to denote the query transformation corresponding to
an input $x=(x_1,\cdots,x_N)$.
Given $i\in [N]$ to the oracle $\mathcal{O}_x$, it returns $x_{i}$.

A quantum computation with $T$ queries is a sequence of unitary transformations
\[
U_0\rightarrow\mathcal{O}_x\rightarrow U_1\rightarrow\mathcal{O}_x\rightarrow\cdots
\rightarrow U_{T-1}\rightarrow\mathcal{O}_x\rightarrow U_{T}
\]
where $U_j$ can be any unitary transformations that do not depend on the input
$x=x_1\cdots x_N$. $\mathcal{O}_x$ are query (oracle) transformations. The oracle
$\mathcal{O}_x$ can be defined as
$\mathcal{O}_x:|i,b,z\rangle\rightarrow|i,b\oplus x_i,z\rangle$, where $\oplus$ is
\emph{exclusive or} operation. Also, we can define $\mathcal{O}_x$ as
$\mathcal{O}_x:|i,b,z\rangle\rightarrow (-1)^{b\cdot x_i}|i,b,z\rangle$, 
where $i$ is the query register, $b$ is the answer register, and $z$ is the working
register.These two definitions of $\mathcal{O}_x$ are equivalent: one query of one type
can be simulated by one query of the other type. The quantum computation are the following
three steps:
\begin{enumerate}
\item[1:] Prepare the initial state to $|0\rangle$.
\item[2:] Then apply $U_0, \mathcal{O}_x,\cdots,\mathcal{O}_x,U_T$.
\item[3:] Measure the final state.
\end{enumerate}

The result of the computation is the rightmost bit of the state obtained by measurement.
The quantum computation computes $f$ with bounded error if, for every $x=(x_1,\cdots,x_N)$
, the probability that the rightmost bit of
$U_TO_xU_{T-1}\cdots O_xU_0|0\rangle$ equals $f(x_1,\cdots,x_N)$ is at least $1-\epsilon$
for some fixed $\epsilon < \frac{1}{2}$. The quantum query complexity of $f$ is the
number of queries needed to compute $f$.

This model can be  extended to functions defined on a larger set or functions having
more than two values.

\subsubsection{Quantum Search}

A search problem in an $n$ elements set $[n]$ is a subset $J\subseteq [n]$ with the
characteristic function $f:[n]\rightarrow\{0,1\}$ such that
\[
f(x) = \left\{\begin{array}{lc}
1,& \mbox{if}\; x\in J,\\
0,& \mbox{otherwise}.
\end{array}\right.
\]
Any $x\in J$ is called a solution of the search problem.

In this paper, we use a generalization of Grover's search algorithm as a quantum
sub-routine, denoted by \emph{GroverAlgorithm}, to determine whether there is any
solution in a search space of size $N$. The quantum sub-routine needs $O(\sqrt{N})$
queries. We state the generalization of Grover's search algorithm as the
following theorem.

\begin{theorem}[see \cite{DBLP:conf/stoc/Grover96, WOL/Boyer98}]
  \label{theorem4groveralgorithm}
  Let $J$ be a search problem in an $n$ elements set $[n]$ and $f$ be the 
  characteristic function of $J$. Given a search space $S\subseteq [n]$ with $|S|=N$,
  determining that whether $J\cap S$ is empty can be done in $O(\sqrt{N})$
  quantum queries to $f$ with probability of at least a constant.
\end{theorem}

\subsubsection{Quantum Query Lower Bounds}
In this paper, we use a quantum adversary method introduced by Ambainis to prove lower
bounds for quantum query complexity.

\begin{theorem}[\textbf{Ambainis's quantum adversary method}\cite{DBLP:conf/stoc/Ambainis00}]
  \label{theorem4quantumadersarymethod}
  Let $f(x_1,\cdots,x_n)$ be a function of $n$ variables with values from a some
  finite set and $X,Y$ be two sets of inputs such that $f(x)\neq f(y)$ if $x\in X$ and
  $y\in Y$. Let $R\subset X\times Y$ be a relation such that
  \begin{enumerate}
  \item For every $x\in X$, there exist at least $m$ different $y\in Y$ such that
    $(x,y)\in R$;
  \item For every $y\in Y$, there exist at least $m'$ different $x\in X$ such that
    $(x,y)\in R$;
  \item For every $x\in X$ and $i\in [n]$, there are at most $l$ different
    $y\in Y$ such that $(x,y)\in R$ and $x_i\neq y_i$;
  \item For every $y\in Y$ and $i\in [n]$, there are at most $l'$ different
    $x\in X$  such that $(x,y)\in R$ and $x_i\neq y_i$;    
  \end{enumerate}
  Then any quantum algorithm computing $f$ uses $\Omega(\sqrt{\frac{mm'}{ll'}})$ queries.
\end{theorem}

\section{Quantum lower bounds and algorithms for matroid property problems}\label{matroidpropertyproblems}

\subsection{A Uniform Lower Bounds For Matroid Properties}

As mentioned in \cite{DBLP:journals/siamcomp/JensenK82}, for a large number of matroid
properties there is no good algorithm for determining whether these properties
holds for general matroids. These properties include uniform matroid, paving matroid,
Eulerian matroid and bipartite matroid decision problems and some calculation
problems, such as finding the girth, circuit number, base number, flat number, hyperplane
number and the size of the largest hyperplane.
We will prove that there is no 
quantum query algorithm with polynomial independence oracles for these problems.

\begin{theorem}
  \label{theorem4nopolymoialalgorithm}
  For the calculation problem  of finding the girth  or the number of circuits (bases, flats, hyperplanes), a quantum algorithm has
  to query the independence oracle at least
  $\Omega(\sqrt{\binom{n}{\lfloor n/2 \rfloor}})$ times.
\end{theorem}
\begin{proof}
  \label{proof4nopolymoialalgorithm}
  Let $V$ be a $n$-set,$r$ be an integer with $0\leq r\leq n$, and $U_{r,n}$ be an uniform
  matroid with rank $r$ on $V$. For a $r$-set $A\subseteq V$, $A^{1}_{r,n}$ is a
  matroid on $V$ with bases all $r$-set excepts $A$.
  We encode every matroid with rank $r$ on $V$ to a  $\binom{n}{r}$-bit
  $0-1$ string(see Matroid Representation in Section \ref{matroidrepresentation}). For every valid
  representation $x\in\{0,1\}^{\binom{n}{r}}$ and any $i\in[\binom{n}{r}]$,
  $x[i]=1$ indicates that a $r$-set which is encoded to the $i$-the bit is a base of
  $\mathcal{M}(x)$.
  Let $X=\{x(U_{r,n})\in\{0,1\}^{\binom{n}{r}}\}$,
  $Y=\{x(A^{1}_{r,n})\in\{0,1\}^{\binom{n}{r}}:\forall A\subseteq V\;\text{with}\; |A|=r \}$.
  Define the following functions from $\{0,1\}^{\binom{n}{r}}$ to a finite set:
  
  $g:\{0,1\}^{\binom{n}{r}}\rightarrow [r+1]$,
  
  $c:\{0,1\}^{\binom{n}{r}}\rightarrow [\binom{n}{r}]$,
  
  $b:\{0,1\}^{\binom{n}{r}}\rightarrow [\binom{n}{r}]$,
  
  $f:\{0,1\}^{\binom{n}{r}}\rightarrow [2^{n}]$,
  
  $h:\{0,1\}^{\binom{n}{r}}\rightarrow [\binom{n}{r}]$,\newline
  which corresponds to finding
  the girth, the number of circuits, the number of bases, the number of flats, the
  number of hyperplanes, the size of the largest hyperplane
  of $\mathcal{M}$, respectively.

  For any function $F\in\{g,c,b,f,h\}$, we have
  $F(x(U_{r,n}))\neq F(x(A^{1}_{r,n}))$.
  Let $R\subseteq X\times Y$ be a relation such that
  \begin{enumerate}
  \item For every $x\in X$, there are $\binom{n}{r}$ different $y\in Y$ such that
    $(x,y)\in R$.
  \item For every $y\in Y$, there is one $x\in X$ such that $(x,y)\in R$.
  \item For every $x\in X$ and $i\in[\binom{n}{r}]$, there is at most one $y\in Y$
    such that $(x,y)\in R$ and $x_i\neq y_i$.
  \item For every $y\in Y$ and $i\in[\binom{n}{r}]$, there is at most one $x\in X$
    such that $(x,y)\in R$ and $x_i\neq y_i$.
  \end{enumerate}
  
  It can be seen that any quantum algorithm
  computing $F$ uses $\Omega(\sqrt{\binom{n}{r}})$ queries.
  When $r=\lfloor n/2\rfloor$, the theorem is proved.
\end{proof}

\begin{theorem}
  \label{theorem4nopolymoialalgorithm2}
  Let $\mathcal{M}=(V,\mathcal{I})$ be a matroid with $|V|=n$. Then any quantum
  algorithm to decide whether $\mathcal{M}$ is uniform or Eulerian has to query
  the independence oracle at least $\Omega(\sqrt{\binom{n}{\lfloor n/2\rfloor}})$ times.
\end{theorem}

\begin{proof}
  \label{proof4nopolymoialalgorithm2}
  Let $r=\lfloor n/2\rfloor$ and $U_{r,n}$ be the uniform matroid with rank $r$ on
  the ground set $V$. We can see that $U_{r,n}$ is not Eulerian, since the size of
  its circuits are $r+1$ implies that any two circuits are intersectant. For any
  $r$-set $A\subseteq V$, let $A^{1}_{r,n}$ be the matroid on $V$ with bases all
  $r$-set except $A$. Similarly, when $n$ is even, let $A^{2}_{r,n}$ be the matroid
  on $V$ with bases all $r$-set except $A$ and $V-A$.
  (It is not difficult to verify that $A^{1}_{r,n}$ and $A^{2}_{r,n}$ are matroids
  with rank $r$ on $V$.)

  In the following, we show that $A^{1}_{r,n}$ and $A^{2}_{r,n}$ are
  Eulerian matroids based on the parity of $n$. When $n$ is odd,
  we can see that $|V-A|=r+1$ and any proper subset of $|V-A|$ is independent in
  $A^{1}_{r,n}$. So $V-A$ is a circuit of $A^{1}_{r,n}$. By the definition of
  $A^{1}_{r,n}$, we also know that $A$ is a circuit of $A^{1}_{r,n}$. Thus we can
  show that $A^{1}_{r,n}$ is Eulerian for the union of disjoint sets $A$ and $V-A$
  is $V$. When $n$ is even, by the definition of $A^{2}_{r,n}$,
  we can see that $A$ and $V-A$ are two disjoint circuits of $A^{2}_{r,n}$. Thus
  $A^{2}_{r,n}$ is also Eulerian.

  We encode every matroid with rank $r$ on $V$ to a $\binom{n}{r}$-bit 0-1 string
  (see Matroid Representation in Section \ref{matroidrepresentation}).
  Let $X=\{x(U_{r,n})\in\{0,1\}^{\binom{n}{r}}\}$,                                           
  $Y_1=\{x(A^{1}_{r,n})\in\{0,1\}^{\binom{n}{r}}:\forall A\subseteq V \text{with} |A|=r\}$,
  $Y_2=\{x(A^{2}_{r,n})\in\{0,1\}^{\binom{n}{r}}:\forall A\subseteq V \text{with} |A|=r\}$.
  Define two Boolean functions
  $u:\{0,\}^{\binom{n}{r}}\rightarrow\{0,1\}$ and
  $e:\{0,\}^{\binom{n}{r}}\rightarrow\{0,1\}$,
  where $u(x)=1$ if and only if $\mathcal{M}(x)$ is uniform and $e(x)=1$ if and
  only if $\mathcal{M}(x)$ is Eulerian. Let $f\in\{u,e\}$. $Y=Y_2$ when $f=e$ and
  $n$ is even, otherwise $Y=Y_1$. We can see that $f(x)\neq f(y)$ for any $x\in X$
  and $x\in Y$. Let $R\subseteq X\times Y$ be a relation such that
  \begin{enumerate}
  \item For every $x\in X$, there are $\binom{n}{r}$($\frac{\binom{n}{r}}{2}$ when
    $f=e$ and $n$ is even) different $y\in Y$ such that $(x,y)\in R$.
  \item For every $y\in Y$, there one $x\in X$ such that $(x,y)\in R$.
  \item For every $x\in X$ and $i\in[\binom{n}{r}]$, there is at most one $y\in Y$
    such that $(x,y)\in R$ and $x_i\neq y_i$.
  \item For every $y\in Y$ and $i\in[\binom{n}{r}]$, there is at most one $x\in X$
    such that $(x,y)\in R$ and $x_i\neq y_i$.
  \end{enumerate}
  It can be seen that any quantum algorithm computing $f$ uses
  $\Omega(\sqrt{\binom{n}{r}})$ queries. When $r=\lfloor n/2\rfloor$,
  the theorem is proved.
\end{proof}

\subsection{Quantum Algorithm for Uniform Matroid Decision Problem}
As mentioned in \cite{CUP/robinson80}, the property
$\lfloor n/2\rfloor$-UNIFORM is a hardest property with respect to the independence
oracle. It requires $\Omega(\binom{n}{\lfloor n/2\rfloor})$ independence 
oracles to determine  whether a matroid is uniform on a ground set with $n$ elements.
We will present a quantum algorithm using
$O(\sqrt{\binom{n}{\lfloor n/2\rfloor}})$ independence oracles to solve this
problem and thus the  algorithm is  asymptotically optimal.

\begin{fact}
  \label{fact4computingrank}
  Let  $\mathcal{M}=(V,\mathcal{I})$  be a matroid with $|V|=n$. There is a greedy
  algorithm to compute the rank of $\mathcal{M}$
  using $O(n)$ independence oracles.
\end{fact}
\begin{proof}
  Let $J=\emptyset$, for every $e\in V$, if $J\cup \{e\}$ is an independent set in
  $\mathcal{M}$, then update $J$ (i.e. $J:=J\cup\{e\}$, otherwise discard $e$).
  This procedure, denoted by \emph{GreedyAlgorithm}, is essentially a greedy algorithm
  with all the elements being the same
  weight. We must traverse all elements in $V$. The fact is proved.
\end{proof}

\begin{theorem}
  \label{theorem4uniformmatroidupperbound}
  Let $\mathcal{M}=(V,\mathcal{I})$ be a matroid with $|V|=n$. There is a quantum
  algorithm to decide whether it is uniform  using
  $O(\sqrt{\binom{n}{\lfloor n/2\rfloor}})$ independence oracles.
\end{theorem}
\begin{proof}
  For determining whether a matroid is uniform, we just need to check the sets
  whose size are equal to its rank.
  Consider the Algorithm \ref{algorithm4uniformmatroid}: Uniform Matroid.
  First, we use Fact \ref{fact4computingrank} to
  compute the rank $r$ of $\mathcal{M}$. Then use GroverAlgorithm to
  determine that whether there is a $S\subseteq V$ with $|S|=r$ such that
  $\mathcal{O}_i(S)=0$ in the set $\{J\subseteq V:|J|=r\}$. If we find such a $S$, we
  can affirmatively conclude that $\mathcal{M}$ is not a uniform matroid. Otherwise
  we repeat GroverAlgorithm several times. Finally, if
  all the results of GroverAlgorithm such that
  $\mathcal{O}_i(S)=1$, it can be shown that $\mathcal{M}$ is an uniform matroid
  with a high probability. The total number of independence oracle required
  is $O(n)+\text{MAX\_REPEAT}\cdot\sqrt{\binom{n}{r}}$. 
  And the number MAX\_REPEAT is a constant. 
  When the rank $r=\lfloor n/2\rfloor$, the algorithm uses 
  $O(\sqrt{\binom{n}{\lfloor n/2\rfloor}})$ independence oracles.
\end{proof}

By Theorem \ref{theorem4uniformmatroidupperbound} and
Theorem \ref{theorem4nopolymoialalgorithm2}, we can
see that the presented algorithm  is asymptotically optimal.
This means that the quantum query complexity of the uniform decision
problem is $\Theta(\binom{n}{\lfloor n/2\rfloor})$.\\

\begin{minipage}{0.95\textwidth}
\begin{lstlisting}[morekeywords={if,int,for,return,while,to},
    emph={GreedyAlgorithm,GroverAlgorithm},emphstyle=\color{blue!50!green},
    caption=\textcolor{gray!20!red}{Uniform Matroid},
   label={algorithm4uniformmatroid},mathescape]
  /*------------------------------------------------------
  Input : Matroid $\mathcal{M}$ accessed by $O_i$ and ground set $V$.
  Output: 1 indicates that $\mathcal{M}$ is uniform, otherwise 0.
  -----------------------------------------------------*/
  int UniformMatroid(M,V)
  {
    int r=GreedyAlgorithm(M,V);
    int k=MAX_REPEAT;
    while(k>0)
    { 
      S=GroverAlgorithm($\{J\subseteq V:|J|=r\}$);
      if ($O_{i}(S)$==0) return 0;
      k=k-1;
    }
    return 1;
  }
\end{lstlisting}
\end{minipage}

\subsection{Quantum Algorithm for Finding Girth}
\begin{theorem}
  \label{theorem4findgirth}
  Let  $\mathcal{M}=(V,\mathcal{I})$  be a matroid with $|V|=n$. There is a quantum
  algorithm to find the girth of  $\mathcal{M}$ using
  $O(\log n\sqrt{\binom{n}{\lfloor n/2\rfloor}})$ independence oracles.
\end{theorem}

\begin{proof}
  Consider the Algorithm \ref{algorithm4computegirth}: Compute Girth, where the
  GreedyAlgorithm is used to compute the rank $r$ of $\mathcal{M}$ and the
  GroverAlgorithm is used to determine whether there is a
  $S\in \{J\subseteq V:|J|=k\}$ such that $\mathcal{O}_{i}(S)=0$ which implies that
  the girth of $\mathcal{M}$ is not greater that $k$. 
  Here we use binary search to find a circuit with girth size. 
  
  Suppose the girth of $\mathcal{M}$ is not $\infty$ and consider the worst case that
  the outer \emph{while} is executed $\log n$ times. 
  Each time the outer \emph{while} is executed, we repeatedly call GroverAlgorithm up to
  MAX\_REPEAT times. The number MAX\_REPEAT is a constant. 
  And the maximum number of independence oracle queried by GroverAlgorithm is
  $\sqrt{\binom{n}{\lfloor n/2\rfloor}}$. The total number of query to independence oracle
  is not more than $O(n)+\text{MAX\_REPEAT}\cdot\log n\sqrt{\binom{n}{\lfloor n/2\rfloor}}$.
  So \emph{ComputeGirth} uses $O(\log n\sqrt{\binom{n}{\lfloor n/2\rfloor}})$ independence queries to compute the girth with probability of at least a constant.
\end{proof}
By Theorem \ref{theorem4findgirth} and Theorem \ref{theorem4nopolymoialalgorithm}, 
we can see that the quantum algorithm is almost optimal.

\begin{minipage}{0.95\textwidth}
\begin{lstlisting}[morekeywords={if,int,for,return,while},
    emph={GreedyAlgorithm,GroverAlgorithm},emphstyle=\color{blue!50!green},
    caption=\textcolor{gray!20!red}{Compute Girth},
    label={algorithm4computegirth},
    mathescape]
/*------------------------------------------------------
  Input : Matroid $\mathcal{M}$ accessed by $O_i$ and ground set $V$.
  Output: The girth of $\mathcal{M}$.
  ------------------------------------------------------*/
  int ComputeGirth(M,V)
  {
    int girth=$\infty$;
    if ($O_{i}(V)$==1) return girth;
    int r=GreedyAlgorithm(M,V);
    int lindex=1, rindex=r+1, k=0;
    while(lindex != rindex)
    {
      girth = (lindex+rindex)/2;
      k=MAX_REPEAT;
      while(k)
      {
        S=GroverAlgorithm($\{J\subseteq V:|J|=girth\}$);
        if ($O_{i}(S)$==0)                          
        {
          rindex=girth;
          break;
        }
        k=k-1;
      }
      if (k==0)
      {
        lindex=girth;
      }
    }
    return girth;
  }
\end{lstlisting}
\end{minipage}

\subsection{Quantum Complexity and Algorithm for Paving Matroid Decision Problem}
Paving matroids are a very important type of matroids in matroid theory.
In the early 1970's, Blackburn, Crapo, and Higg \cite{AMS/Blackburn73} noticed
that most of the matroids on a ground set of up to 8 elements are paving matroids.
Crapo and Rota \cite{MIT/Crapo70} suggested that perhaps paving matroids 
``would actually predominate in any asymptotic enumeration of geometries''.
Mayhew et al. \cite{DBLP:journals/ejc/MayhewNWW11}
gave a conjecture, ``Asymptotically, almost every matroid is paving''.
Here we give an almost optimal quantum algorithm to determine whether a matroid
is paving.

\begin{fact}
  \label{fact4pavingmatroid}
  Given a matroid $\mathcal{M}=(V,\mathcal{I})$ with rank $r\geq 2$, if there exists a
  circuit $C$ in $\mathcal{M}$ with $|C|\leq r-2$. Then there must be a dependent set
  $J$ with $|J|=r-1$ in $\mathcal{M}$.
\end{fact}
\begin{proof}  
  This fact is obvious. Let $X\subseteq V\setminus C$ with $|X|=r-1-|C|$, then
  $C\subseteq J=C\cup X$ is a dependent set of $\mathcal{M}$ with $|J|=r-1$.
\end{proof}

\begin{theorem}
  \label{theorem4pavingmatroidupperbound}
  Let  $\mathcal{M}=(V,\mathcal{I})$  be a matroid with $|V|=n$. There is a quantum
  algorithm to decide whether it is paving using $O(\sqrt{\binom{n}{\lfloor n/2\rfloor}})$
  independence oracles.
\end{theorem}
\begin{proof}
  Consider the Algorithm \ref{algorithm4pavingmatroid}: {\sc Paving Matroid}.
  By Fact \ref{fact4pavingmatroid}, we can see that a matroid $\mathcal{M}$ with rank
  $r\geq 2$ is a paving matroid if and only if every the $(r-1)$-set
  is an independent set. First we use Fact \ref{fact4computingrank} to compute the rank
  $r$ of $\mathcal{M}$. Then we use GroverAlgorithm to determine whether there is a
  a $S\subseteq V$ with $|S|=r-1$ such that $\mathcal{O}_{i}(S)=0$ in the set
  $\{J\subseteq V:|J|=r-1\}$. If we find such a $S$, we
  can affirmatively conclude that $\mathcal{M}$ is not a paving matroid. Otherwise
  we repeat GroverAlgorithm several times. Finally, if all the results of
  GroverAlgorithm such that $\mathcal{O}_i(S)=1$, it can be shown that $\mathcal{M}$
  is a paving matroid with a high probability.
  The total number of independence query oracle required
  is $O(n)+\text{MAX\_REPEAT}\cdot\sqrt{\binom{n}{r-1}}$. 
  And the number MAX\_REPEAT is a constant.
  When $r-1=\lfloor n/2\rfloor$,
  the total number of independence query oracles used is
  $O(\sqrt{\binom{n}{\lfloor n/2\rfloor}})$.
\end{proof}

\begin{minipage}{0.95\textwidth}
\begin{lstlisting}[morekeywords={if,int,for,return,while},
    emph={GreedyAlgorithm,GroverAlgorithm},emphstyle=\color{blue!50!green},
    caption=\textcolor{gray!20!red}{Paving Matroid},
    label={algorithm4pavingmatroid},
    mathescape]
 /*------------------------------------------------------
  Input : Matroid $\mathcal{M}$ accessed by $O_i$ and ground set $V$.
  Output: 1 indicates that $\mathcal{M}$ is paving, otherwise 0.
  -----------------------------------------------------*/
  int PavingMatroid(M,V)
  {
    int r=GreedyAlgorithm(M,V);
    int k=MAX_REPEAT;
    while(k>0)
    {
      S=GroverAlgorithm($\{J\subseteq V:|J|=r-1\}$);
      if ($O_{i}(S)$==0) return 0;  
      k=k-1;
    }
    return 1;
  }
\end{lstlisting}
\end{minipage}

\begin{fact}
  \label{fact4differenceinstring}
  Given a ground set $V$ with $|V|=n$, an integer $r$ with $1\leq r\leq n$ and
  $A,C\subseteq V$ with $|A|=|C|=r-1$. Define $J_{r},J_{A},J_{C}$ as Notations in Section 
  \ref{matroidnotation}. Let
  $\mathcal{B}_{A}=J_{r}-J_{A}$,
  $\mathcal{B}_{C}=J_{r}-J_{C}$.
  We use a $\binom{n}{r}$ bit $0-1$ string $x(\mathcal{M})$ encodes a matroid
  $\mathcal{M}$ with $\text{rank}(\mathcal{M})=r$. For $i\in[\binom{n}{r}]$,
  $x(\mathcal{M})[i]=1$ indicates that some $r$-set be a base of $\mathcal{M}$.
  For any different $A$ and $C$, the two matroids determined by $\mathcal{B}_{A}$
  and $\mathcal{B}_{C}$ as the collection of bases, denoted by $\mathcal{M}_{A}$ and
  $\mathcal{M}_{C}$, then there is at most one $i\in[\binom{n}{r}]$ such that
  $x(\mathcal{M}_{A})[i]=x(\mathcal{M}_{C})[i]=0$.
\end{fact}
\begin{proof}
  For any two different $(r-1)$-set $A$ and $C$, we have $0\leq |A\cap C| \leq r-2$.
  $|A\cap C|< r-2$ implies that for any $e\in V-A$ and $e'\in V-C$, we have
  $A_e\neq C_{e'}$. Thus for any $i\in[\binom{n}{r}]$, $x(\mathcal{M}_{A})[i]$
  and $x(\mathcal{M}_{C})[i]$ will not be 0 at the same time.
  $|A\cap C|=r-2$ implies that there exists one and only one $r$-set
  (that is $A\cup C$) be the superset of $A$ and $C$ with cardinality $r$.  
  The fact is proved.
\end{proof}
\begin{theorem}
  \label{theorem4pavingmatroidlowerbound}
  Let  $\mathcal{M}=(V,\mathcal{I})$  be a matroid with $|V|=n$.  Then any quantum
  query algorithm  to decide whether $\mathcal{M}$ is paving requires at least
  $\Omega(\sqrt{\binom{n}{\lfloor n/2\rfloor}/n})$ independence oracles.
\end{theorem}

\begin{proof}
  Given an integer $r$ with $1\leq r \leq n$ and a subset $A\subseteq V$ with $|A|=r-1$,
  $e\in V-A$, define $J_{r}, J_{A}$ as Notations in Section  \ref{matroidnotation}.
  Let $\mathcal{B}_{A}=J_{r}-J_{A}$.
  Then there is a matroid on $V$ that can be determined by $\mathcal{B}_{A}$ as the
  collection of bases, denoted by $\mathcal{M}_{A}$.
  Furthermore, $\mathcal{M}_A$ is not a
  paving matroid (because
  $A\notin \mathcal{I}(\mathcal{M}_{A})$ and $\text{rank}(\mathcal{M}_{A})=r > |A|$).

  We encode every matroid with rank $r$ on $V$ to a different $\binom{n}{r}$-bit
  $0-1$ string(see Matroid Representation \ref{matroidrepresentation}). Let  
  $X=\{x(U_{r,n})\in\{0,1\}^{\binom{n}{r}}\}$,
  $Y=\{x(\mathcal{M}_{A})\in\{0,1\}^{\binom{n}{r}}:\forall A\subseteq V\;\text{with}\;|A|=r-1)\}$.
  Define a Boolean function $f:\{0,1\}^{\binom{n}{r}}\rightarrow \{0,1\}$, $f(x)=1$ 
  if and only if the matroid $\mathcal{M}(x)$ is a paving matroid.
  By the definition of $X$ and $Y$, it is easy to verify that every
  $x\in X$ is a paving matroid and every $y\in Y$ is not a paving matroid. So
  $X$,$Y$ are two sets of inputs such that
  $f(x)\neq f(y)$ if $x\in X$ and $y\in Y$. Let $R\subset X_{r}\times Y_{r}$
  be a relation such that
  \begin{enumerate}
  \item For every $x\in X$, there exist $\binom{n}{r-1}$ different $y\in Y$ such
    that  $(x,y)\in R$.
  \item For every $y\in Y$, there exists one $x\in X$ such that $(x,y)\in R$.
  \item For every $x\in X$ and $i\in[\binom{n}{r}]$, there are at most 
    $\binom{r}{r-1}$ (known from Fact \ref{fact4differenceinstring}) different
    $y\in Y$ such that $(x,y)\in R$ and $x_i\neq y_i$.
  \item For every $y\in Y$ and $i\in[\binom{n}{r}]$, there is at most one
    $x\in X$ such that $(x,y)\in R$ and $x_i\neq y_i$.
  \end{enumerate}  
  
  It can be seen 
  that any quantum algorithm computing $f$ uses
  $\Omega(\sqrt{\binom{n}{r-1}/r})$ queries. When $r-1=\lfloor n/2\rfloor$,
  we obtain the lower bound 
  $\Omega(\sqrt{\binom{n}{\lfloor n/2\rfloor}/n})$.
\end{proof}

By Theorem \ref{theorem4pavingmatroidupperbound} and
Theorem \ref{theorem4pavingmatroidlowerbound}, one can see that our quantum algorithm is
almost optimal.

\subsection{Trivial and Loopless Decision Problems}
\begin{theorem}
  \label{theorem4trivialandloopless}
  The quantum query complexity of trivial and loopless decision problems are
  $\Theta(\sqrt{n})$.
\end{theorem}

\begin{proof}
  From Definitions \ref{definition4freeandtrivialmatroid} and
  \ref{definition4looplessmatroid},
  we can see that a matroid is a trivial matroid if and only if all the
  singleton in $V$ are dependent, and a matroid is loopless if and only if all the
  singleton in $V$ are independent. In other words, for the trivial(loopless) matroid
  decision problem, if there is an $e\in V$ such that $\mathcal{O}_{i}(\{e\})=1(0)$ implies
  the matroid is not a trivial(loopless) matroid. So these two decision problems are
  essentially the unordered search problem in the set $V$. 
  By \cite{DBLP:conf/stoc/Grover96, WOL/Boyer98, PRA/Zalka99},
  the quantum query complexity of the trivial and loopless decision problems are
  $\Theta(\sqrt{n})$.
\end{proof}

\section{Conclusion}
In this paper, we discussed quantum speedup and limitations on matroid property problems, assuming that a matriod
can be accessed through the independence oracle. We obtained lower bounds on the
quantum query complexity for the calculation problem of finding the girth or the number of circuits (bases, flats, hyperplanes)
of a matroid,  and for the decision problem of deciding whether a matroid is
 uniform  or Eulerian. These lower bounds imply that there is no
polynomial-time quantum algorithm for these problems. We also presented  quantum
algorithms with potential quadratic  speedup  over classical ones for some of these problems
and the algorithms are asymptotically optimal or almost optimal. 

There are a large number of matroids in a ground set with $n$ elements.
This gives us many potential possibilities. There are some interesting questions
worthy of further consideration.

\begin{enumerate}
\item Are there  matroid problems for which quantum algorithms can show
  superpolynomial speedup over their counterparts?
  If yes, can we characterize the properties owned by these problems?
\item Robinson and Welsh\cite{CUP/robinson80} have shown that the capabilities of oracles given by 
  different definitions are very different. What about in the quantum case?
\item In this paper, we simply use Grover's algorithm to solve some relatively
  simple problems.
  What about other problems?
  Especially the matroid intersection problem and the matroid partition problem.
\end{enumerate}
\section*{Acknowledgement}
This work was supported by the National Natural Science Foundation of China (Grant No. 61772565), the Guangdong Basic and Applied Basic Research Foundation (Grant No. 2020B1515020050), the Key Research and Development project of Guangdong Province (Grant No. 2018B030325001)

\bibliography{main-full}

\end{document}